\documentclass{toc}
\usepackage{graphicx}
\usepackage{verbatim}

\tocdetails{%
   volume=3, number=3, year=2007, firstpage=45,
   received={July 28, 2006},
   published={March 28, 2007},
   doi={10.4086/toc.2007.v003a003}
}

\newcommand{\eps}{\epsilon}

\newcommand{\Vvar}[2]{v_{#1}^{(#2)}}
\newcommand{\Uvar}[2]{u_{#1}^{(#2)}}
\newcommand{\Wvar}[2]{w_{#1}^{(#2)}}
\newcommand{\Evar}[2]{e_{#1}^{(#2)}}
\newcommand{\Yvar}[2]{y_{#1}^{(#2)}}

\def\myminus{{\mbox{\protect\raisebox{0.1em}{-}}}}

\newcommand{\true}{\mathsf{True}}
\newcommand{\false}{\mathsf{False}}

\newcommand{\unsat}{\overline{\mathsf{SAT}}}
\newcommand{\NP}{\mathsf{NP}}

\newcommand{\coNP}{\mathsf{coNP}}

\newcommand{\YES}{\mathsf{YES}}
\newcommand{\NO}{\mathsf{NO}}
\newcommand{\alphabeta}{\mathsf{[\alpha,\beta]}}
\newcommand{\Bsat}{{\mathsf{B}}}
\newcommand{\Dsat}{\mathsf{D}}
\newcommand{\cnf}{\Psi}

\newcommand{\sat}{\mathsf{SAT}}
\newcommand{\threesat}{\rm{3}\myminus\mathsf{SAT}}

\newcommand{\gapsat}{\mathsf{\forall\exists}\myminus\mathsf{SAT}}
\newcommand{\gapthreesat}{\mathsf{\forall\exists}\myminus{\rm{3}}\myminus\mathsf{SAT}}
\newcommand{\gapKthreesat}{{\mathsf{(\forall\exists)}^r}\myminus{\rm{3}}\myminus{\mathsf{SAT}}}

\newcommand{\gapKthreesatB}{{(\forall\exists)}^r\myminus{\rm{3}}\myminus{\mathsf{SAT}}\myminus\Bsat}
\newcommand{\EgapKthreesatB}{\mathsf{\exists(\forall\exists)}^r\myminus{\rm{3}}\myminus\mathsf{SAT}\myminus\Bsat}

\newcommand{\gapKDsatB}{{(\forall\exists)}^r\myminus{\rm{\Dsat}}\myminus{\mathsf{SAT}}\myminus\Bsat}
\newcommand{\gapthreesatB}{\mathsf{\forall\exists}\myminus{\rm{3}}\myminus\mathsf{SAT}\myminus\Bsat}
\newcommand{\gapDsat}{\mathsf{\forall\exists}\myminus\rm{\Dsat}\myminus\mathsf{SAT}}
\newcommand{\gapDsatB}{\mathsf{\forall\exists}\myminus\rm{\Dsat}\myminus\mathsf{SAT}\myminus\Bsat}

\newcommand{\gapDsatBX}{\mathsf{\forall\exists}\myminus\rm{\Dsat}\myminus\mathsf{SAT}\myminus\Bsat_\forall}

\newcommand{\CNF}{\mathsf{CNF}}

\runningauthor{I.~Haviv, O.~Regev, and A.~Ta-Shma}

\runningtitle{On the Hardness of Satisfiability with Bounded
Occurrences}

\copyrightauthor{Ishay Haviv, Oded Regev, and Amnon Ta-Shma}

\begin{document}

\begin{frontmatter}[classification=float]

\title{On the Hardness of Satisfiability with Bounded
Occurrences in the Polynomial-Time Hierarchy}

\tocpdftitle{On the Hardness of Satisfiability with Bounded
Occurrences in the Polynomial-Time Hierarchy}

\tocpdfauthor{Ishay Haviv, Oded Regev, Amnon Ta-Shma}

\author[ishay]{Ishay Haviv}

\author[oded]{Oded Regev\thanks{Supported by an Alon Fellowship, by the Binational Science
Foundation, by the Israel Science Foundation, and
 by the EU Integrated Project QAP.}}

\author[amnon]{Amnon Ta-Shma\thanks{Supported by the Binational Science Foundation, by the Israel Science Foundation, and
 by the EU Integrated Project QAP.}}

\tockeywords{satisfiability, polynomial-time hierarchy, expander graphs, superconcentrator graphs}

\begin{abstract}
In 1991, Papadimitriou and Yannakakis gave a reduction implying the
$\NP$-hardness of approximating the problem $\threesat$ with bounded
occurrences. Their reduction is based on expander graphs. We present
an analogue of this result for the second level of the
polynomial-time hierarchy based on superconcentrator graphs. This
resolves an open question of Ko and Lin (1995) and should be useful
in deriving inapproximability results in the polynomial-time
hierarchy.

More precisely, we show that given an instance of $\gapthreesat$ in
which every variable occurs at most $\Bsat$ times (for some absolute
constant $\Bsat$), it is $\Pi_2$-hard to distinguish between the
following two cases: $\YES$ instances, in which for any assignment
to the universal variables there exists an assignment to the
existential variables that satisfies {\em all} the clauses, and
$\NO$ instances in which there exists an assignment to the universal
variables such that any assignment to the existential variables
satisfies at most a $1-\eps$ fraction of the clauses. We also
generalize this result to any level of the polynomial-time
hierarchy.
\end{abstract}

\tocams{03D15, 68Q17}
\tocacm{F.1.3}

\end{frontmatter}

\section{Introduction}\label{sec:intro}

In the problem $\gapthreesat$, given a $3\myminus\CNF$ formula we
have to decide whether for any assignment to a set of universal
variables $X$ there exists an assignment to a set of existential
variables $Y$, such that the formula is satisfied. Here, by a
$3\myminus\CNF$ formula we mean a conjunction of clauses where each
clause is a disjunction of at most $3$ literals. This problem is a
standard $\Pi_2$-complete problem. We denote the corresponding gap
problem by $\gapthreesat[1-\eps_1,1-\eps_2]$ where $0 \leq \eps_2 <
\eps_1 \leq 1$. This is the problem of deciding whether for any
assignment to the universal variables there exists an assignment to
the existential variables such that at least a $1-\eps_2$ fraction
of the clauses are satisfied, or there exists an assignment to the
universal variables such that any assignment to the existential
variables satisfies at most a $1-\eps_1$ fraction of the clauses.
The one-sided error gap problem $\gapthreesat[1-\eps,1]$ is
$\Pi_2$-hard for some $\eps>0$, as was shown in~\cite{KoLin94}. This
problem has the perfect completeness property, \ie, in $\YES$
instances it is possible to satisfy \emph{all} the clauses.

In this paper we consider a restriction of $\gapthreesat$, known as
$\gapthreesatB$. Here, each variable appears at most $\Bsat$ times
where $\Bsat$ is some constant. In~\cite{KoLin95}, Ko and Lin showed
that $\gapthreesatB[1-\eps_1,1-\eps_2]$ is $\Pi_2$-hard for some
constants $\Bsat$ and $0<\eps_2<\eps_1<1$. Our main result is that
the problem is still $\Pi_2$-hard for some $\eps_1>0$ with
$\eps_2=0$, \ie, with perfect completeness. This solves an open
question given in~\cite{KoLin95}.

\begin{theorem}\label{thm:gapEthreesatB}
The problem $\gapthreesatB[1-\eps,1]$ is $\Pi_2$-hard for some
constants $\Bsat$ and $\eps>0$. Moreover, this is true even when the
number of literals in each clause is exactly $3$.
\end{theorem}

We note that the problem remains $\Pi_2$-hard even if the number of
occurrences of universal variables is bounded by $2$ and the number
of occurrences of existential variables is bounded by $3$. As we
will explain later, these are the least possible constants for which
the problem is still $\Pi_2$-hard unless the polynomial-time
hierarchy collapses. We believe that \thmref{thm:gapEthreesatB} is
useful for deriving $\Pi_2$-hardness results, as well as $\Pi_2$
inapproximability results. In fact, \thmref{thm:gapEthreesatB} was
crucial in a recent proof that the covering radius problem on
lattices with high norms is $\Pi_2$-hard~\cite{CRPhardS06}.
Moreover, using \thmref{thm:gapEthreesatB}, one can simplify the
proof that the covering radius on codes is $\Pi_2$-hard to
approximate~\cite{GuruswamiMR04}.

At a very high level, the proof is based on the following ideas.
First, one can reduce the number of occurrences of existential
variables by an expander construction in much the same way as was
done by Papadimitriou and Yannakakis~\cite{PapaYanna91}. The main
difficulty in the proof is in reducing the number of occurrences of
universal variables: If we duplicate universal variables (as is
usually done in order to reduce the number of occurrences), we have
to deal with inconsistent assignments to the new universal variables
(this problem shows up in the completeness proof). The approach
taken by Ko and Lin~\cite{KoLin95} is to duplicate universal
variables and to add existential variables on top of the universal
variables. Their construction, in a way, enables the existential
variables to override inconsistent assignments to the universal
variables. Unfortunately, it seems that this technique cannot
produce instances with perfect completeness. In our approach we also
duplicate the universal variables, but instead of using them
directly in the original clauses, we use a superconcentrator-based
gadget, whose purpose is intuitively to detect the majority among
the duplicates of a universal variable. Crucially, this gadget
requires only a constant number of occurrences of each universal
variable.

The rest of the paper is organized as follows. \secref{sec:prem}
provides some background about satisfiability problems in the second
level of the polynomial-time hierarchy and about some explicit
expanders and superconcentrators. In \secref{sec:sat} we prove
\thmref{thm:gapEthreesatB}. \secref{sec:numberOfOcc} discusses the
least possible value of $\Bsat$ for which the problem remains
$\Pi_2$-hard. In \secref{sec:HigherLevelsOfTheHierarchy} we
generalize our main theorem to any level of the polynomial-time
hierarchy.

\section{Preliminaries}\label{sec:prem}

\subsection{$\Pi_2$ satisfiability problem}

A $\Dsat\myminus\CNF$ formula over a set of variables is a
conjunction of clauses where each clause is a disjunction of {\em at
most} $\Dsat$ literals. Each literal is either a variable or its
negation. A clause is satisfied by a Boolean assignment to the
variables if it contains at least one literal that evaluates to
$\true$.

For any reals $0\leq\alpha<\beta\leq 1$ and positive integer
$\Dsat>0$, we define:
\begin{definition}[$\gapDsat{\alphabeta}$]\label{def:Pi2SAT}
An instance of $\gapDsat{\alphabeta}$ is a $\Dsat\myminus\CNF$
Boolean formula $\cnf(X,Y)$ over two sets of variables. We refer to
variables in $X$ as universal variables and to those in $Y$ as
existential variables. In $\YES$ instances, for every assignment to
$X$ there exists an assignment to $Y$ such that at least a $\beta$
fraction of the clauses are satisfied. In $\NO$ instances, there
exists an assignment to $X$ such that for every assignment to $Y$ at
most an $\alpha$ fraction of the clauses are satisfied.
\end{definition}

\begin{sloppypar}
The problem $\gapDsat{\alphabeta}$ is the basic approximation
problem in the second level of the polynomial-time hierarchy
(see~\cite{SchaeferUmans02a,SchaeferUmans02b} for a recent survey on
the topic of completeness and hardness of approximation in the
polynomial-time hierarchy). We also define some additional variants
of the above problem. For any $\Bsat\geq 1$ the problem
$\gapDsatB{\alphabeta}$ is defined similarly except that each
variable occurs at most $\Bsat$ times in $\cnf$. In the instances of
the problem $\gapDsatBX{\alphabeta}$, the bound $\Bsat$ on the
number of occurrences applies only to the universal variables (as
opposed to all variables).
\end{sloppypar}

In~\cite{KoLin95} it was shown that
$\gapthreesatB[1-\eps_1,1-\eps_2]$ is $\Pi_2$-hard for some $\Bsat$
and some  $0<\eps_2<\eps_1<1$. As already mentioned, in
\secref{sec:sat} we show that it is $\Pi_2$-hard even for some
$\Bsat$, $\eps_1>0$ and $\eps_2=0$.

\subsection{Expanders and superconcentrators}

In this subsection, we gather some standard results on explicit
constructions of expanders and superconcentrators (where by {\em
explicit} we mean constructible in polynomial time). The first shows
the existence of certain regular expanders.

\begin{lemma}[\cite{LubotzkyPhSa88,Margulis88}]\label{lem:expandersStep2}
There exists a universal constant $C_1$ such that for any integer
$n$, there is an explicit $14$-regular graph $G=(V,E)$ with $n\leq
|V|\leq C_1 n$ vertices, such that any nonempty set $S\subset V$
satisfies $|E(S,\overline{S})|>\min(|S|,|\overline{S}|)$.
\end{lemma}

For the second, we need to define the notion of a superconcentrator.

\begin{definition}[$n$-superconcentrator]
A directed acyclic graph $G=(U \cup V \cup W,E)$ where $U$ denotes a
set of $n$ inputs (\ie, vertices with indegree $0$) and $V$ denotes
a set of $n$ outputs (\ie, vertices with outdegree $0$) is an {\em
$n$-superconcentrator} if for any subset $S$ of $U$ and any subset
$T$ of $V$ satisfying $|S|=|T|$, there are $|S|$ vertex-disjoint
directed paths in $G$ from $S$ to $T$.
\end{definition}

The explicit construction of sparse superconcentrators has been
extensively studied. Gabber and Galil~\cite{GabberGa81} were the
first to give an explicit expander-based construction of
$n$-superconcentrator with $O(n)$ edges. Alon and
Capalbo~\cite{AlonCapalbo} presented the most economical known
explicit $n$-superconcentrators, in which the number of edges is
$44n+O(1)$. Their construction is based on a modification of the
well-known construction of Ramanujan graphs by Lubotzky, Phillips
and Sarnak~\cite{LubotzkyPhSa88} and by Margulis~\cite{Margulis88}.
The following theorem of~\cite{AlonCapalbo} summarizes some of the
properties of their graphs.

\begin{theorem}[\cite{AlonCapalbo}]\label{thm:superconcentrators}
There exists an absolute constant $k>0$ for which the following
holds. For any $n$ of the form $k \cdot 2^l$ $(l\geq 0)$ there
exists an explicit $n$-superconcentrator $H=(U \cup V \cup W,E)$
with $|E|=44n+O(1)$ and all of whose vertices have indegree and
outdegree at most $11$.
\end{theorem}

In our reduction, we use a slight modification of the
superconcentrator in \thmref{thm:superconcentrators}. This graph is
described in the following claim (see \figref{fig:supercon} for an
illustration of the construction).

\begin{figure}[ht]
  \begin{center}
    \includegraphics[width=4in]{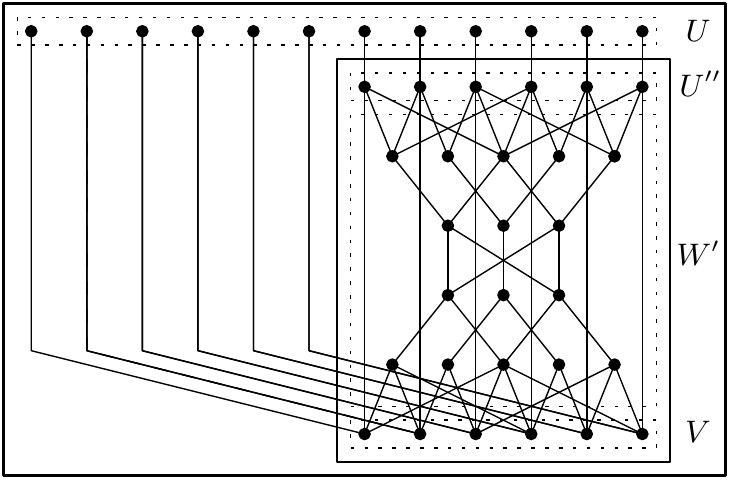}  %
  \end{center}
  \caption{The graph $G^{(6)}$. All edges are directed downwards. The marked subgraph is
 a 6-super\-con\-cen\-trator (but not necessarily the one from~\cite{AlonCapalbo}).}
  \label{fig:supercon}
\end{figure}

\begin{claim}\label{claim:oursuperconcentrators}
There exist absolute constants $c$ and $d$ for which the following
holds. For any natural $n \geq 1$ there exists an explicit directed
acyclic graph $G^{(n)}=(U \cup V \cup W,E)$ with a set $U$ of $2n$
inputs (\ie, vertices with indegree $0$) with outdegree $1$ and a
set $V$ of $n$ outputs (\ie, vertices with outdegree $0$), such that
for any subset $S$ of $U$ of size $|S|=n$ there are $n$
vertex-disjoint directed paths from $S$ to $V$. Moreover, $|E| \leq
cn$ and all indegrees and outdegrees in $G^{(n)}$ are bounded by
$d$.
\end{claim}

\begin{proof}
Fix some $n \geq 1$. By \thmref{thm:superconcentrators} there exists
an explicit $n_0$-superconcentrator $H'=(U' \cup V' \cup W',E')$ for
some $n+k \leq n_0 < 2(n+k)$ where $k$ is the constant from
\thmref{thm:superconcentrators}, such that $|E'|=44n_0+O(1)$ and all
its indegrees and outdegrees are bounded by $11$. Denote by
$U''=\{u''_1,\ldots,u''_n\}$ and by $V=\{v_1,\ldots,v_n\}$ arbitrary
subsets of $U'$ and $V'$ of size exactly $n$.

In order to construct the graph $G^{(n)}$ we add to the graph $H'$
the $2n$ vertices $U=\{u_1,\ldots,u_{2n}\}$ and $2n$ edges. The
input set of the graph $G^{(n)}$ is $U$, and the output set of
$G^{(n)}$ is $V$. For each $i \in \{1,\ldots,n\}$ we add the
directed edges $(u_i,u''_i)$ and $(u_{i+n},v_i)$. In other words, we
add to the graph two matchings of size $n$: the first between the
vertex sets $\{u_{1},\ldots,u_{n}\}$ and $U''$, and the second
between $\{u_{n+1},\ldots,u_{2n}\}$ and $V$.

It is easy to see that our graph satisfies the required properties
for large enough absolute constants $c$ and $d$. Let $S \subseteq U$
be of size $n$, and define $S_1 = S \cap \{u_i : 1 \leq i \leq n\}$
and $S_2= S \cap \{u_i : n+1 \leq i \leq 2n\}$. We show that there
exist $n$ vertex-disjoint paths from $S$ to $V$. According to our
construction, the vertices of $S_2$ have paths of length $1$ to
their neighbors in $V$. So it suffices to show that the vertices of
$S_1$ have vertex-disjoint paths to the $n-|S_2|=|S_1|$ remaining
vertices of $V$. According to the property of $H'$, there exist
vertex-disjoint paths in $G^{(n)}$ between the neighbors of $S_1$ in
$U''$ and the $n-|S_2|$ vertices of $V$. Combining these paths
together with the matching edges between $S_1$ and $U''$ completes
the proof.
\end{proof}

\section{\texorpdfstring{Hardness of approximation for $\gapthreesatB$}{Hardness of approximation for Forall-Exists-3-SAT-B}}
\label{sec:sat}

In this section we prove \thmref{thm:gapEthreesatB}. The proof is by
reduction from the problem $\gapthreesat[1-\eps,1]$, which was shown
to be $\Pi_2$-hard for some $\eps>0$ in~\cite{KoLin94}. The
reduction is performed in three steps. The first step is the main
one, and it is here that we present our new superconcentrator-based
construction. The remaining two steps are standard (see for
example~\cite{VaziraniBook} and~\cite{AroraLund96}) and we include
them mainly for completeness. We remark that these two steps are
also used in~\cite{KoLin95}.

\begin{description}
  \item[Step 1:] Here we reduce the number of occurrences of
  each universal variable to at most some constant $\Bsat$. As a
  side effect, the size of the clauses grows from being at most $3$ to being at
  most $\Dsat$, where $\Dsat$ is some constant. More precisely, we establish that there exist absolute constants $\Bsat$, $\Dsat$ and $\eps>0$ such that
  the problem $\gapDsatBX[1-\eps,1]$ is $\Pi_2$-hard.
  \item[Step 2:] Here we reduce the number of occurrences of
  the existential variables to some constant $\Bsat$. Notice that we
  must make sure that this does not affect the number of occurrences
  of the universal variables. More precisely, we show that there exist absolute constants $\Bsat$, $\Dsat$ and $\eps>0$ such
  that the problem $\gapDsatB[1-\eps,1]$ is $\Pi_2$-hard.
  \item[Step 3:] Finally, we modify the formula such that
  the size of the clauses is exactly $3$. Clearly, we must make
  sure that the number of occurrences of each variable remains constant.
  This would complete the proof of \thmref{thm:gapEthreesatB}.
\end{description}

\subsection{Step 1}

Before presenting the first step we offer some intuition. In order
to make the number of occurrences of the universal variables
constant we replace their occurrences by new and distinct
existential variables. In detail, assume $x$ is a universal variable
that occurs $\ell$ times in an instance $\cnf$ of
$\gapthreesat[1-\eps,1]$. For such a variable we construct the graph
$G^{(\ell)}=(U \cup V \cup W,E)$ given in
\clmref{claim:oursuperconcentrators} and identify its $\ell$ output
vertices $V$ with the $\ell$ new existential variables. In addition,
we associate a universal variable with each of the $2 \ell$ vertices
of $U$, and an existential variable with each vertex in $W$ and also
with each edge in $E$. We add clauses that verify that in the
subgraph of $G^{(\ell)}$ given by the edges with value $\true$,
there are $\ell$ vertex-disjoint paths from $U$ to $V$ (and hence
each vertex in $V$ has one incoming path). We also add clauses that
verify that if an edge has value $\true$ then both its endpoints
must have the same value. This guarantees that each variable in $V$
gets the value of one of the variables in $U$. Completeness follows
because for any assignment to $U$, we can assign all the variables
in $V$ to the same value by connecting them to those variables in
$U$ that get the more popular assignment (recall that $|U|=2|V|$ and
the properties given in \clmref{claim:oursuperconcentrators}). For
the proof of soundness, we show that if all the $U$ variables are
assigned the same value, then all the $V$ variables should also be
assigned this value.

\subsubsection{The reduction}
The proof is by reduction from the problem $\gapthreesat[1-\eps,1]$
which is $\Pi_2$-hard for some constant $\eps>0$ as shown
in~\cite{KoLin94}. Let $\cnf(X,Y)$ be a $3\myminus\CNF$ Boolean
formula with $m$ clauses over the set of variables $X\cup Y$, where
$X=\{x_1,\ldots,x_{|X|}\}$ is the set of universal variables, and
$Y=\{y_1,\ldots,y_{|Y|}\}$ is the set of existential variables. The
reduction constructs a formula $\cnf'(X',Y')$ over $X'\cup Y'$. The
number of occurrences in $\cnf'$ of each universal variable from
$X'$ will be bounded by an absolute constant $\Bsat$, and the number
of literals in each clause will be at most $\Dsat$. In fact, these
constants are $\Bsat = 2$ and $\Dsat=d+1$, where $d$ is given in
\clmref{claim:oursuperconcentrators}.

For each universal variable $x_i\in X$ denote by $\ell_i$ the number
of its occurrences in the formula $\cnf$, and apply
\clmref{claim:oursuperconcentrators} to obtain the graph
$G_i=G^{(\ell_i)}=(U_i \cup V_i \cup W_i,E_i)$. Recall that the
maximum degree (indegree and outdegree) of these graphs is bounded
by some constant $d$ and that the number of edges in $G_i$ is
bounded by $c \cdot \ell_i$ for some constant $c$. Denote the vertex
sets of $G_i$ by
\[V_i=\{\Vvar{1}i,\ldots,\Vvar{\ell_i}i\}, \qquad
U_i=\{\Uvar{1}i,\ldots,\Uvar{2\ell_i}i\}, \qquad \text{and}\qquad
W_i=\{\Wvar{1}i,\ldots,\Wvar{|W_i|}i\}\enspace,
\]
and its edge set by $E_i=\{\Evar{1}i,\ldots,\Evar{|E_i|}i\}$. The set of
existential variables in $\cnf'$ is 
\[
Y'=\left(\bigcup_{i=1}^{|X|}{(V_i\cup W_i
    \cup E_i)}\right)\cup Y\enspace.
\]
The set of universal variables in $\cnf'$ is $X'=\bigcup_{i=1}^{|X|}{U_i}$.

\begin{figure}
  \begin{center}
    \includegraphics[scale=1.2]{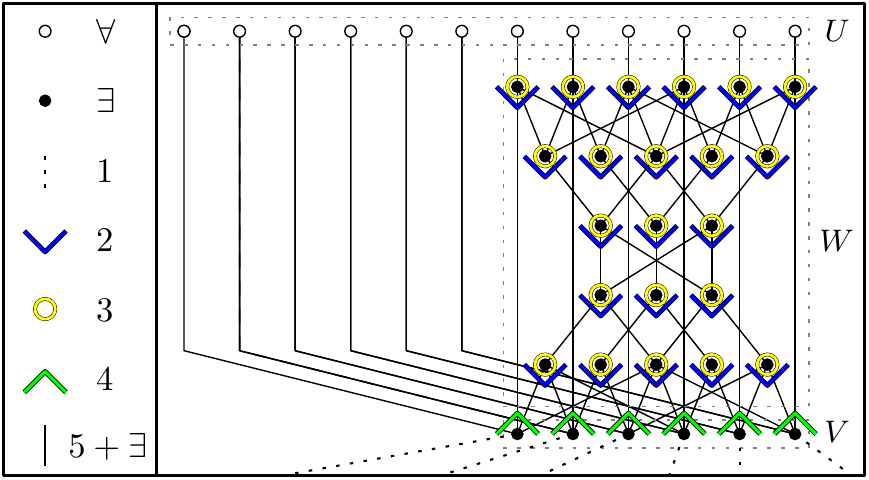}  %
  \end{center}
  \caption{An illustration of the reduction for the case $\ell=6$.}
  \label{fig:reduction}
\end{figure}

The clauses of $\cnf'$ are divided into the following five types
(see \figref{fig:reduction}).
\begin{enumerate}
  \item\textbf{Major clauses:} These clauses are obtained from
  clauses of the formula $\cnf$, by replacing the $j$th occurrence of
  the universal variable $x_i$ with the variable 
  \[
  \Vvar{j}i \in V_i
  \]
  for $1\leq i\leq |X|$, $1\leq j\leq\ell_i$. The number of clauses of this type is
  $m$.
  \item \textbf{Outdegree clauses:} These clauses
  verify that among the directed edges leaving a vertex in $G_i$, at most one has value $\true$.
  For each vertex $w$, we add the clause
  \[
  (\lnot\Evar{j_1}{i}\lor\lnot\Evar{j_2}{i})
  \]
  for each pair of edges $\Evar{j_1}{i},\Evar{j_2}{i}$ leaving $w$.
  Each such clause is duplicated $d^2$ times. The number of clauses of
  this type is at most $\ell_i\cdot c \cdot d^2 {{d}\choose{2}}$ for each $i$.
  \item \textbf{Flow clauses:} These clauses
  verify for any vertex $\Wvar{j}i \in W_i$ that if at least one of its outward edges
  has value $\true$ then there exists also an edge entering $\Wvar{j}i$ with value $\true$.
  This is done by adding a clause of the form
  \[
  (\lnot\Evar{j'}{i}\lor\Evar{j_1}{i}\lor\cdots\lor\Evar{j_{d'}}{i})
  \]
  for each $\Evar{j'}{i}$ leaving $\Wvar{j}i$ where
  $\Evar{j_1}{i},\ldots,\Evar{j_{d'}}{i}$ are all the $0 \leq d' \leq
  d$ edges entering $\Wvar{j}i$. The number of clauses of this type is
  at most $c \cdot\ell_i$ for each $i$.
  \item \textbf{$V$-degrees clauses:} These clauses verify
  that each vertex $\Vvar{j}{i}$ has at least one
  incident edge with $\true$ value. This is done by
  adding one clause of the form
  \[
  (\Evar{j_1}{i}\lor\cdots\lor\Evar{j_{d'}}{i})
  \]
  where $\Evar{j_1}{i},\ldots,\Evar{j_{d'}}{i}$ are the $d' \leq d$ edges
  incident to $\Vvar{j}i$. The number of clauses of this type is $\ell_i$
  for each $i$.
  \item \textbf{Edge consistency clauses:} For each edge
  $\Evar{j}i \in E_i$ do the following. Let $\Wvar{j_1}{i},\Wvar{j_2}{i} \in U_i \cup V_i \cup W_i$ be its endpoints.
  Add the two clauses 
  \[(\lnot\Evar{j}{i}\lor\Wvar{j_1}{i}\lor\lnot\Wvar{j_2}{i}) \qquad \text{and} \qquad
  (\lnot\Evar{j}{i}\lor\lnot\Wvar{j_1}{i}\lor\Wvar{j_2}{i})\enspace,
  \]
  which check that if the value of $\Evar{j}i$ is
  $\true$, then $\Wvar{j_1}{i}$ and
  $\Wvar{j_2}{i}$ have the same truth value. The number of
  clauses of this type is at most $2 c \ell_i$ for each $i$.
\end{enumerate}

Note that each clause contains at most $\Dsat=d+1$ literals. Using
$\sum_i{\ell_i}\leq 3m$, the number of clauses in $\cnf'$, which we
denote by $m'$, is at most $O(mc\cdot (d^4+1))\leq C\cdot m$ for
some absolute constant $C$. Moreover, the number of occurrences of
each universal variable is exactly $2$, because universal variables
appear only in clauses of type $(5)$ and vertices in the $U_i$ have
outdegree $1$. This completes the construction of $\cnf'$.

\subsubsection{Completeness}
Our goal in the completeness proof is to show that if $\cnf(X,Y)$ is
a $\YES$ instance of $\gapthreesat[1-\eps,1]$, then for any
assignment to $X'$, there is an assignment to $Y'$ that satisfies
all the $m'$ clauses in $\cnf'(X',Y')$. Let $t'$ be an arbitrary
assignment to the universal variables $X'$. Recall that $X'$ is the
union $\bigcup_{i=1}^{|X|}{U_i}$. We define an assignment $t$ to $X$
based on the majority of the assignments given by $t'$. More
formally,
\begin{eqnarray*}\label{eq:assignmentdef}
t(x_i)=\left\{
         \begin{array}{ll}
           \true, & \text{if $|\{j:t'(\Uvar{j}i)=\true\}|\geq\ell_i$}\enspace, \\
           \false, & \text{otherwise}.
         \end{array}
       \right.
\end{eqnarray*}
By the assumption on the original formula $\cnf(X,Y)$, the
assignment $t$ can be extended to $X\cup Y$, in a way that satisfies
all the clauses in $\cnf(X,Y)$. Let us extend the assignment $t'$ to
the existential variables 
\[
 Y'=\Bigl(\bigcup_{i=1}^{|X|}{(V_i\cup W_i\cup E_i)}\Bigr)\cup Y\enspace.
\] First, let the assignment $t'$ give the same values
as $t$ for the variables in $Y$. For each $i$ denote by $S_i
\subseteq U_i$ a set of vertices from $U_i$ of size $|S_i|=\ell_i$
in which every variable has value $t(x_i)$. There exists such a set
according to the definition of $t$. By
\clmref{claim:oursuperconcentrators} there are $\ell_i$
vertex-disjoint directed paths in $G_i$ from $S_i$ to $V_i$. We
define $t'(\Evar{j}i)$ to be $\true$ if $\Evar{j}i$ appears in one
of these paths and $\false$ otherwise. In addition, $t'$ gives the
value $t(x_i)$ to all variables in $V_i \cup W_i$.

We now check that the assignment $t'$ satisfies all clauses in
$\cnf'$. The assignment to the variables in $V_i$ is $t(x_i)$. Since
the variables $Y$ are also assigned according to $t$, all clauses of
type $(1)$ are satisfied. The paths given by
\clmref{claim:oursuperconcentrators} are vertex-disjoint. In
particular, every vertex has at most one outward edge assigned to
$\true$, so all clauses of type $(2)$ are satisfied too. Moreover,
if at least one of the edges leaving a vertex $w \in W_i$ has value
$\true$ then there exists also a directed edge with value $\true$
entering $w$. Therefore, the clauses of type $(3)$ are satisfied.
The number of paths in $G_i$ is $\ell_i$, so there is one path
reaching every vertex in $V_i$. This means that the clauses of type
$(4)$ are satisfied too. Finally, our assignment gives the value
$t(x_i)$ to all variables in $S_i \cup V_i \cup W_i$. In particular,
each edge assigned to $\true$ has both its endpoints with the same
value. Thus, the clauses of type $(5)$ are satisfied, as required.

\subsubsection{Soundness}
In the soundness proof we assume $\cnf(X,Y)$ is a $\NO$ instance of
$\gapthreesat[1-\eps,1]$. We will show the existence of an
assignment to $X'$ for which any assignment to $Y'$ satisfies at
most $(1-\eps')m'$ clauses of $\cnf'(X',Y')$ for
$\eps'={\eps}/{C}$, and hence the theorem will follow.

Let $t$ be an assignment to $X$ such that every extension of $t$ to
$X\cup Y$ satisfies at most $(1-\eps)m$ clauses in $\cnf(X,Y)$.
Define an assignment $t'$ to $X'$ in which every variable
$\Uvar{j}i$ has the value $t(x_i)$. Extend $t'$ to an assignment to
$X'\cup Y'$ in an arbitrary way. Our goal in the following is to
show that the number of clauses satisfied by $t'$ is at most
$(1-\eps')m'$. We start with the following two claims.

\begin{claim}\label{claim:satisfying2}
Let $t'$ be an assignment to $X' \cup Y'$ as above. Then $t'$ can be
modified to an assignment $t''$ that satisfies every clause of type
$(2)$ and satisfies at least as many clauses as $t'$ satisfies.
\end{claim}

\begin{proof}
We obtain $t''$ by performing the following modification to $t'$ for
each $i$: For each variable in $W_i$, if it has more than one
outward edge assigned to $\true$ by $t'$, $t''$ assigns $\false$ to
all its outward edges. Since we only modify variables in $E_i$,
clauses of type $(1)$ are not affected. Moreover, since we only set
edges to $\false$, we do not decrease the number of satisfied
clauses of type $(5)$. We might, however, reduce the number of
satisfied clauses of types $(3)$ and $(4)$ by at most $d^2$ for each
variable (at most $d$ for each out-neighbor of the vertex). On the
other hand, the corresponding clause of type $(2)$ is satisfied by
$t''$, and by the duplication, this amounts to at least $d^2$
additional satisfied clauses. In total, the number of clauses
satisfied by $t''$ is at least the number of clauses satisfied by
$t'$, and the claim follows.
\end{proof}

\begin{claim}\label{claim:Kunsatisfied}
Let $t'$ be an assignment to $X' \cup Y'$ that satisfies all clauses
of type $(2)$. Denote by $k$ the number of vertices $\Vvar{j}i \in
\bigcup_{l}{V_l}$ satisfying $t'(\Vvar{j}i) \neq t(x_i)$, where $t$
is the assignment to $X$ as above. Then at least $k$ clauses of
types $(3)$, $(4)$ or $(5)$ are unsatisfied by $t'$.
\end{claim}

\begin{proof}
Fix some $i$. It suffices to show that to each vertex $\Vvar{j}i$
satisfying $t'(\Vvar{j}i) \neq t(x_i)$ we can assign in a one-to-one
fashion a clause of type $(3)$, $(4)$ or $(5)$ which is not
satisfied by $t'$. To show this let $G'$ be the subgraph of $G_i$
given by the edges assigned to $\true$ by $t'$. Let $A_j$ be the set
of vertices that have a directed path in $G'$ to $\Vvar{j}{i}$.
Since clauses of type $(2)$ are all satisfied by $t'$, the sets
$A_j$ are pairwise disjoint. Fix some $1 \leq j \leq \ell_i$ such
that $t'(\Vvar{j}{i}) \neq t(x_i)$. Since $G_i$ is acyclic, $A_j$
contains a vertex $u$ whose indegree in $G'$ is $0$. If $u$ is in
$U_i$ then at least one of the clauses of type $(5)$ on the path
from $u$ to $\Vvar{j}{i}$ is unsatisfied by $t'$, because $t'(u) =
t(x_i)$ whereas $t'(\Vvar{j}i) \neq t(x_i)$. Otherwise at least one
of the clauses of types $(3)$ and $(4)$ is unsatisfied by $t'$.
Therefore, we see that the number of clauses of type $(3)$-$(5)$
unsatisfied by $t'$ is at least the number of vertices $\Vvar{j}i$
satisfying $t'(\Vvar{j}i) \neq t(x_i)$.
\end{proof}

Recall that $t'$ is an assignment to $X' \cup Y'$ that assigns every
variable $\Uvar{j}i$ to $t(x_i)$. We have to show that $t'$
satisfies at most $(1-\eps')m'$ clauses in $\cnf'$. By
\clmref{claim:satisfying2} we can assume that $t'$ satisfies all
clauses of type $(2)$ in $\cnf'$.

Now, we define an assignment $t''$ to $X' \cup Y'$ as follows. For
each $i$, let $S_i$ be an arbitrary subset of $U_i$ of size
$\ell_i$. We know that there exist $\ell_i$ directed vertex-disjoint
paths from $S_i$ to $V_i$ in $G_i$. The assignment $t''$ assigns all
the $\Evar{j}i$ in these paths to $\true$ and all other $\Evar{j}i$
to $\false$. Moreover, $t''$ gives all variables in $U_i \cup V_i
\cup W_i$ the value $t(x_i)$. Finally, we define $t''$ on $Y$ to be
identical to $t'$. Notice that in $t''$ all clauses of type
$(2)$-$(5)$ are satisfied. Denote by $k$ the number of the variables
$\Vvar{j}{i}$ satisfying $t'(\Vvar{j}{i}) \neq t(x_i)$. Then the
number of type $(1)$ clauses satisfied by $t''$ is smaller than that
of $t'$ by at most $k$. Moreover, $t'$ satisfies all clauses of type
$(2)$, so by \clmref{claim:Kunsatisfied} at least $k$ clauses of
type $(3)$-$(5)$ are unsatisfied by $t'$. In total, the number of
clauses satisfied by $t''$ is at least the number of clauses
satisfied by $t'$.

Finally, by our assumption on $\cnf$ and on $t$ we get that at least
$\eps m$ clauses of type $(1)$ are not satisfied by $t''$. So the
number of satisfied clauses is at most $m'-\eps m \leq (1-\eps')m'$,
as required.

\subsection{Step 2}

With Step $1$ proven, we now apply an idea of~\cite{PapaYanna91} to
show that there are absolute constants $\Bsat$ and $\eps>0$, for
which the problem $\gapDsatB[1-\eps,1]$ is $\Pi_2$-hard. This proof
uses the expander graphs from \lemref{lem:expandersStep2}.

\paragraph{The reduction:} Consider the $\Pi_2$-hard problem
$\gapDsatBX[1-\eps',1]$ for some $\eps'>0$. Let $\cnf(X,Y)$ be an
instance of this problem. For every existential variable $y_i\in Y$
($1\leq i\leq |Y|$) denote by $n_i$ the number of the occurrences of
$y_i$ in $\cnf$. Assuming $n_i$ is large enough, consider the graph
$G_i=(V_i,E_i)$ given by \lemref{lem:expandersStep2} for $n_i$, with
$n_i\leq|V_i|\leq C_1 n_i$ (if $n_i$ is not large enough, we do not
need to modify this variable). Label the vertices of $G_i$ with
$|V_i|$ new distinct existential variables
$Y_i=\{\Yvar{1}i,\ldots,\Yvar{|V_i|}i\}$. We construct a new Boolean
formula $\cnf'(X,Y')$ over the universal variables in $X$ and the
existential variables in $Y'=\bigcup_{i=1}^{|Y|}{Y_i}$. First, for
each $1\leq i\leq |Y|$ replace the occurrences of $y_i$ by $n_i$
distinct variables of $Y_i$. Second, for each edge
$(\Yvar{j}i,\Yvar{j'}i)$ in $G_i$, add to $\cnf$ the two clauses
\[
(\lnot\Yvar{j}i\lor\Yvar{j'}i) \qquad\text{and}\qquad (\Yvar{j}i\lor\lnot\Yvar{j'}i)\enspace,
\]
which are both satisfied if and only if the variables
$\Yvar{j}i,\Yvar{j'}i$ have the same value. The number of clauses in
$\cnf'$ is linear in $\sum_{i}{n_i}\leq \Dsat m$. Notice, that the
number of occurrences of \emph{each} variable in $\cnf'$ is bounded
by a constant.

\paragraph{Correctness:} Let $\cnf(X,Y)$, an $m$ clause formula, be a
$\YES$ instance, \ie, for every assignment to $X$ there exists an
assignment to $Y$ such that every clause in $\cnf$ is satisfied.
Clearly, for any assignment to $X$ there exists an assignment to
$Y'$ which satisfies all the clauses in $\cnf'$, because we can set
the $Y_i$ variables the value of $y_i$ in $\cnf$. Now , assume
$\cnf$ is a $\NO$ instance, so there is an assignment $t$ to $X$
such that for any assignment to $Y$ at least $\eps' m$ clauses are
unsatisfied in $\cnf$. Let $t'$ be an arbitrary extension of $t$ to
$X\cup Y'$. If for some $1 \leq i \leq |Y|$, $t'$ does not assign to
all the $Y_i$ variables the same value for some $1\leq i\leq |Y|$,
it is possible to improve the number of satisfied clauses by setting
all the $Y_i$ variables to the majority vote of $t'$ on $Y_i$.
Indeed, denote by $S_i$ the set of variables in $Y_i$ that were
assigned by $t'$ to $\true$. This modification reduces the number of
satisfied clauses by at most $\min(|S_i|,|\overline{S_i}|)$, but
satisfies at least $|E(S_i,\overline{S_i})|$ unsatisfied consistency
clauses. \lemref{lem:expandersStep2} states that
$|E(S_i,\overline{S_i})|>\min(|S_i|,|\overline{S_i}|)$, so this
modification improves the number of satisfied clauses. Hence, we can
assume that for each $1 \leq i \leq |Y|$, $t'$ assigns to all the
$Y_i$ variables the same value for each $1\leq i \leq |Y|$. Thus, by
the assumption on $\cnf$ we conclude that $t'$ does not satisfy at
least $\eps' m$ clauses, meaning at least an ${\eps'}/{\Dsat}$
fraction of the clauses is not satisfied. Defining
$\eps= {\eps'}/{\Dsat}$ completes the proof.

\subsection{Step 3}

This subsection completes the proof of \thmref{thm:gapEthreesatB} by
showing a reduction that modifies the size of the clauses to exactly
$3$.

\paragraph{The reduction:} Let $\cnf(X,Y)$ be an instance of
$\gapDsatB[1-\eps',1]$ with $m$ clauses. We transform $\cnf$ into a
formula $\cnf(X',Y')$, whose clauses are of size exactly $3$, as
follows. For each clause of size $1$, like $(a)$, we add a new
universal variable $z$ and replace it by $(a\lor z\lor z)$.
Similarly, for each clause of size $2$, like $(a\lor b)$, we add a
new universal variable $z$ and replace it by $(a\lor b\lor z)$. Now
consider a clause $C=(u_1\lor u_2\lor\cdots\lor u_r)$ of size $r>3$,
where the $u_i$ are literals. For each such clause introduce $r-3$
new and distinct existential variables $z_1,\ldots,z_{r-3}$ and
replace $C$ in the formula $\cnf$ by the clauses of $C'$,
\[
C'=(u_1\lor u_2\lor z_1)\land(\lnot z_1\lor u_3\lor z_2)\land\cdots\land(\lnot z_{r-4}\lor u_{r-2}\lor
z_{r-3})\land(\lnot z_{r-3}\lor u_{r-1}\lor u_r)\enspace.
\]
The number of the clauses in $\cnf'$ is at most $\Dsat m$.
Obviously, the number of occurrences of each variable remains the
same, and the newly added variables appear either once or twice.

\paragraph{Correctness:} It is easy to see that if $\cnf$ is a $\YES$
instance then so is $\cnf'$ and that if $\cnf$ is a $\NO$ instance,
then there exists an assignment to $X'$ such for any assignment
$Y'$, at least $\eps' m$ of the clauses of $\cnf'(X',Y')$ are
unsatisfied. So for $\eps={\eps'}/{\Dsat}$ we get the desired
result.

\section{On the number of occurrences}\label{sec:numberOfOcc}

The output of the reduction of \secref{sec:sat} is a formula in
which every universal variable occurs at most twice and every
existential variable occurs at most $\Bsat$ times for some constant
$\Bsat$. By performing a transformation similar to the one in Step
$2$ with the graphs of \lemref{lem:expandersStep2} replaced by
directed cycles, the number of occurrences of each existential
variable can be made at most $3$ (see for example Theorem 10.2, Part
1 in~\cite{AroraLund96}). This implies that if we allow each
universal variable to occur at most twice and each existential
variable to occur at most $3$ times, the problem remains
$\Pi_2$-hard. Here, we show that $2$ and $3$ are the best possible
constants (unless the polynomial-time hierarchy collapses).

First note that whenever a universal variable occurs only once in a
formula, we can remove it without affecting the formula. Hence, if
each universal variable occurs at most once, the problem is in $\NP$
and thus is not $\Pi_2$-hard, unless the polynomial-time hierarchy
collapses.

Moreover, if we allow every existential variable to occur at most
twice, the problem lies in $\coNP$ and is thus unlikely to be
$\Pi_2$-hard. Given an assignment to the universal variables $X$,
the formula $\cnf(X,Y)$ becomes a $\sat$ formula in which each
variable appears at most twice. Checking satisfiability of such
formulas can be done in polynomial time~\cite{Tovey84}. Indeed,
variables that appear only once and those that appear twice with the
same sign can be removed from the formula together with the clauses
that contain them. This means that we are left with a $\sat$ formula
in which each variable appears once as a positive literal and once
as a negative one. So consider the bipartite graph $H=(A \cup B,E)$
in which $A$ is the set of clauses of $\cnf$ and $B$ is the set of
its existential variables. We connect by an edge a clause in $A$ to
a variable in $B$ if the clause contains the variable. Notice that
there exists a matching in $H$ that saturates $A$ if and only if the
formula is satisfiable. The existence of such a matching can be
checked easily in polynomial time. Therefore $\gapsat$ restricted to
instances in which every existential variable occurs at most twice
is in $\coNP$.

\section{Extension to higher levels of the hierarchy}\label{sec:HigherLevelsOfTheHierarchy}

As one might expect, \thmref{thm:gapEthreesatB} can be generalized
to any level of the polynomial-time hierarchy. In this section, we
describe in some detail how this can be done. Our aim is to prove
the following theorem (the problems below are the natural extension
of $\gapthreesat$ to higher levels of the hierarchy;
see~\cite{KoLin94}).

\begin{theorem}\label{thm:generalization}
For any $r \geq 1$ there exists an $\eps>0$ such that
$\gapKthreesatB[1-\eps,1]$ is $\Pi_{2r}$-complete and
$\EgapKthreesatB[1-\eps,1]$ is $\Sigma_{2r+1}$-complete (where
$\Bsat$ is some absolute constant). Moreover, this is true even when
the number of literals in each clause is exactly $3$.
\end{theorem}

For convenience, we present the proof only for the even levels of
the hierarchy ($\Pi_{2r}$). The case of odd levels is almost
identical.

Our starting point is a result of~\cite{KoLin94}, which says that
for any $r \geq 1$ there exists an $\eps>0$ such that
$\gapKthreesat[1-\eps,1]$ is $\Pi_{2r}$-complete. As in
\secref{sec:sat}, the proof proceeds in three steps. In the first we
reduce the number of occurrences of universal variables. In the
second we reduce the number of occurrences of existential variables.
Finally, in the third step we modify the formula such that the size
of each clause is exactly $3$.

\subsection{Step 1}

In this step we show that for any $\eps>0$ there exists an $\eps'>0$
such that $\gapKthreesat[1-\eps,1]$ reduces to
$\gapKDsatB_{\forall}[1-\eps',1]$ for some absolute constants
$\Dsat, \Bsat$ (where the latter problem is a restriction of the
former to instances in which each universal variable appears at most
$\Bsat$ times). In more detail, given a $3$-$\CNF$ formula $\cnf$ on
variable set $X_1 \cup Y_1 \cup \cdots \cup X_r \cup Y_r$, we show
how to construct a $\Dsat$-$\CNF$ formula $\cnf'$ on variable set
$X'_1 \cup Y'_1 \cup \cdots \cup X'_r \cup Y'_r$ in which each
universal variable appears at most $\Bsat$ times, and whose size is
linear in the size of $\cnf$, such that
\begin{align}
&\max_{t_{X_1}} \min_{t_{Y_1}} \cdots \max_{t_{X_r}} \min_{t_{Y_r}}
    {\unsat(\cnf,t_{X_1},t_{Y_1},\ldots,t_{X_r},t_{Y_r})} \nonumber \\ \label{eq:goal}
  & \qquad = \max_{t_{X'_1}} \min_{t_{Y'_1}} \cdots
\max_{t_{X'_r}} \min_{t_{Y'_r}}
{\unsat(\cnf',t_{X_1'},t_{Y_1'},\ldots,t_{X_r'},t_{Y_r'})}\enspace,
\end{align}
where $\unsat$ denotes the number of {\em unsatisfied} clauses in a
formula for a given assignment. It is easy to see that this is
sufficient to establish the correctness of the reduction. Moreover,
it can be verified that in Step 1, \secref{sec:sat} we proved
Equation~\eqref{eq:goal} for the case $r=1$.

Before describing the reduction, we note that in Step 1,
\secref{sec:sat}, the only property of the original formula that we
used is that flipping the value of an occurrence of a variable can
change the number of satisfied clauses by at most one. This leads us
to the following lemma, whose proof was essentially given already in
Step 1, \secref{sec:sat}.

\begin{lemma}\label{lem:hierarchy_univ}
For any $\ell \geq 1$ there exists a $k \geq \ell$ and a
$\Dsat$-$\sat$ formula $\Phi(x_1,\ldots,x_{2\ell},y_1,\ldots,y_k)$
(for some absolute constant $\Dsat$) on $2\ell+k$ variables of size
$O(\ell)$ in which each of the first $2\ell$ variables appears at
most twice such that the following holds. For any integer-valued
function $f$ on $\ell$ Boolean variables with the property that
flipping any one variable changes the value of $f$ by at most one,
we have that
$$ \max_x f(x,\ldots,x) = \max_{x_1,\ldots,x_{2\ell}} \min_{y_1,\ldots,y_k} (f(y_1,\ldots,y_\ell) + \unsat(\Phi,x_1,\ldots,x_{2\ell},y_1,\ldots,y_k))\enspace,$$
where $x,x_1,\ldots,x_{2\ell},y_1,\ldots,y_k$ are Boolean variables.
\end{lemma}

Using this lemma we can now describe our reduction. We are given a
$3$-$\CNF$ formula $\cnf$ on variable set $X_1 \cup Y_1 \cup \cdots
\cup X_r \cup Y_r$. We perform the following modifications for each
universal variable $x$. Let $i$ be such that $x \in X_i$ and $\ell$
be the number of times $x$ occurs in $\cnf$. Let $k$ and $\Phi$ be
as given by \lemref{lem:hierarchy_univ}. First, we replace $x \in
X_i$ with $2 \ell$ new variables $x_1,\ldots,x_{2\ell} \in X_i$ and
add $k$ new variables $y_1,\ldots,y_k$ to $Y_i$. Next, we replace
the $\ell$ occurrences of $x$ with $y_1,\ldots,y_\ell$. Finally, we
append $\Phi(x_1,\ldots,x_{2\ell},y_1,\ldots,y_k)$ to the formula.
Let $\cnf'$ be the resulting formula and $X'_1 \cup Y'_1 \cup \cdots
\cup X'_r \cup Y'_r$ be the resulting variable set. This completes
the description of the reduction.

Clearly, each universal variable in $\cnf'$ appears at most twice,
and moreover, the size of $\cnf'$ is linear in that of $\cnf$.
Therefore it remains to prove Equation~\eqref{eq:goal}. We do this by
showing that for each universal variable, the modifications we
perform leave the expression in Equation~\eqref{eq:goal} unchanged. So
let $\cnf$ be an arbitrary formula on some variable set $X_1 \cup
Y_1 \cup \cdots \cup X_r \cup Y_r$, and let $x \in X_i$ be a
universal variable with $\ell$ occurrences. It can be seen that our
goal is to show that\footnote{We remark that the fact that we write
$\max_{t_{X_i\setminus \{x\}}} \max_{x}$ as opposed to $\max_{x}
\max_{t_{X_i\setminus \{x\}}}$ will be crucial when we apply
\lemref{lem:hierarchy_univ}, as this prevents an additional
quantifier alternation.}
\begin{align*}
&\max_{t_{X_1}} \min_{t_{Y_1}} \cdots \max_{t_{X_i\setminus \{x\}}}
\max_{x} \min_{t_{Y_i}}  \cdots \max_{t_{X_r}} \min_{t_{Y_r}}
    {g(t_{X_1},t_{Y_1},\ldots,t_{X_i \setminus \{x\}},x,\ldots,x,t_{Y_i},\ldots,t_{X_r},t_{Y_r})} \\
  & \qquad = \max_{t_{X_1}} \min_{t_{Y_1}} \cdots \max_{t_{X_i\setminus \{x\}}} \max_{x_1,\ldots,x_{2\ell}} \min_{y_1,\ldots,y_k} \min_{t_{Y_i}}  \cdots \max_{t_{X_r}} \min_{t_{Y_r}} \\
    & \qquad \qquad ({g(t_{X_1},t_{Y_1},\ldots,t_{X_i \setminus \{x\}},y_1,\ldots,y_\ell,t_{Y_i},\ldots,t_{X_r},t_{Y_r})} +
       \unsat(\Phi,x_1,\ldots,x_{2\ell},y_1,\ldots,y_k))\enspace,
\end{align*}
where $g$ denotes the number of unsatisfied clauses in $\cnf$ under
the given assignment to all variables except $x$ and to all
occurrences of $x$, and $k$ and $\Phi$ are as in
\lemref{lem:hierarchy_univ}. Clearly it suffices to prove this
equality for any fixed setting to the variables quantified before
$x$, \ie,
\begin{align*}
&\max_{x} \min_{t_{Y_i}}  \cdots \max_{t_{X_r}} \min_{t_{Y_r}}
    {g(t_{X_1},t_{Y_1},\ldots,t_{X_i \setminus \{x\}},x,\ldots,x,t_{Y_i},\ldots,t_{X_r},t_{Y_r})} \\
  & \qquad = \max_{x_1,\ldots,x_{2\ell}} \min_{y_1,\ldots,y_k} \min_{t_{Y_i}}  \cdots \max_{t_{X_r}} \min_{t_{Y_r}} \\
    & \qquad \qquad ({g(t_{X_1},t_{Y_1},\ldots,t_{X_i \setminus \{x\}},y_1,\ldots,y_\ell,t_{Y_i},\ldots,t_{X_r},t_{Y_r})} +
       \unsat(\Phi,x_1,\ldots,x_{2\ell},y_1,\ldots,y_k))\enspace,
\end{align*}
but this follows from \lemref{lem:hierarchy_univ}.

We conclude that $\gapKDsatB_{\forall}[1-\eps,1]$ is $\Pi_{2r}$-hard
for some $\eps>0$.

\subsection{Step 2}

In this step we show that for any $\eps>0$ there exists an $\eps'>0$
such that $\gapKDsatB_{\forall}[1-\eps,1]$ reduces to
$\gapKDsatB[1-\eps',1]$ for some absolute constants $\Dsat, \Bsat$.
The following lemma is the analogue of \lemref{lem:hierarchy_univ}
for existential variables, and its proof essentially appeared
already in Step 2, \secref{sec:sat}.

\begin{lemma}\label{lem:hierarchy_exist}
For any large enough $\ell$ there exists a $2$-$\sat$ formula
$\Phi(y_1,\ldots,y_{\ell})$ on $\ell$ variables of size $O(\ell)$ in
which each variable appears at most $\Bsat$ times (for some absolute
constant $\Bsat$) such that the following holds. For any
integer-valued function $f$ on $\ell$ Boolean variables with the
property that flipping any one variable changes the value of $f$ by
at most one, we have that
$$ \min_y f(y,\ldots,y) = \min_{y_1,\ldots,y_\ell} (f(y_1,\ldots,y_\ell) + \unsat(\Phi,y_1,\ldots,y_\ell))\enspace,$$
where $y,y_1,\ldots,y_\ell$ are Boolean variables.
\end{lemma}

The reduction is as follows. We are given a $\Dsat$-$\CNF$ formula
$\cnf$ on variable set $X_1 \cup Y_1 \cup \cdots \cup X_r \cup Y_r$.
We perform the following modifications for each existential variable
$y$. Let $i$ be such that $y \in Y_i$ and $\ell$ be the number of
times $y$ occurs in $\cnf$. Let $\Phi$ be as given by
\lemref{lem:hierarchy_exist}. First, we replace $y \in Y_i$ with
$\ell$ variables $y_1,\ldots,y_{\ell} \in Y_i$. Next, we replace the
$\ell$ occurrences of $y$ with $y_1,\ldots,y_\ell$. Finally, we
append $\Phi(y_1,\ldots,y_{\ell})$ to the formula. This completes
the description of the reduction. The proof of correctness is
similar to the previous one and uses \lemref{lem:hierarchy_exist}.

\subsection{Step 3}

To complete the proof of \thmref{thm:generalization} we now modify
the formula so that the number of literals in each clause is exactly
$3$. Given a formula $\cnf$ on variable set $X_1 \cup Y_1 \cup
\cdots \cup X_r \cup Y_r$ we apply the modification of Step 3,
\secref{sec:sat}. We add the new existential variables to $Y_r$ and
the new universal variables to $X_r$. The proof of correctness is
easy and is omitted.

\subsection*{Acknowledgements}
We thank Ker-I Ko for sending us a copy of~\cite{KoLin95}. Some of
the early ideas that eventually led us to the construction of
\secref{sec:sat} were obtained while the second author was working
on~\cite{GuruswamiMR04} together with Daniele Micciancio and
Venkatesan Guruswami. We also thank two anonymous referees for their
helpful comments.

\bibliographystyle{tocplain}
\bibliography{v003a003}

\begin{tocauthors}
\begin{tocinfo}[ishay]
  Ishay Haviv \tocabout \\
  School of Computer Science \\
  Tel Aviv University, Tel Aviv, Israel
\end{tocinfo}
\begin{tocinfo}[oded]
  Oded Regev \tocabout \\
  Assistant professor \\
  School of Computer Science \\
  Tel Aviv University, Tel Aviv, Israel\\
  \url{http://www.cs.tau.ac.il/~odedr}
\end{tocinfo}
\begin{tocinfo}[amnon]
  Amnon Ta-Shma \tocabout \\
  Assistant professor \\
  School of Computer Science \\
  Tel Aviv University, Tel Aviv, Israel\\
  \url{http://www.cs.tau.ac.il/~amnon}
\end{tocinfo}
\end{tocauthors}

\begin{tocaboutauthors}
  \begin{tocabout}[ishay] {\sc Ishay Haviv} is a a graduate student at
    the \href{http://www.cs.tau.ac.il/}{School of Computer Science},
    \href{http://www.tau.ac.il/}{Tel Aviv University}, under the
    supervision of \href{http://www.cs.tau.ac.il/~odedr/}{Oded Regev}.
    He is interested in theoretical computer science, especially the
    complexity of lattice problems.  He also loves animals and acts
    for their rights.
\end{tocabout}
\begin{tocabout}[oded] {\sc Oded Regev} graduated from
  \href{http://www.tau.ac.il/}{Tel Aviv University} in 2001 under the
  supervision of \href{http://www.math.tau.ac.il/~azar/}{Yossi Azar}.
  Before joining Tel Aviv University, he spent two years as a
  postdoctoral fellow at the \href{http://www.ias.edu/}{Institute for
    Advanced Study}, Princeton, and one year at the
  \href{http://www.berkeley.edu/}{University of California, Berkeley}.
  His research interests include quantum computation, computational
  aspects of lattices, and other topics in theoretical computer
  science. He also enjoys hiking, running, and photography.
\end{tocabout}
\begin{tocabout}[amnon] {\sc Amnon Ta-Shma} graduated from the
  \href{http://www.huji.ac.il/}{Hebrew University} in 1996 under the
  supervision of \href{http://www.cs.huji.ac.il/~noam/}{Noam Nisan}.
  Before joining \href{http://www.tau.ac.il/}{Tel Aviv University}, he
  spent three years as a postdoctoral fellow at the
  \href{http://www.icsi.berkeley.edu/}{International Computer Science
    Institute}, Berkeley and the
  \href{http://www.berkeley.edu/}{University of California, Berkeley}.
  His research interests include the role of randomness in
  computation, quantum computation, and other topics in theoretical
  computer science. He also enjoys his family and thanks them for
  their love.
\end{tocabout}

\end{tocaboutauthors}

\end{document}